\documentclass[aps,prl,twocolumn,superscriptaddress,longbibliography]{revtex4-2}
\usepackage{amsmath,amssymb,amsthm}
\usepackage{graphicx}
\usepackage{bm}
\usepackage{bbm}
\usepackage{physics}          
\usepackage[colorlinks=true,linkcolor=blue,citecolor=blue,urlcolor=magenta]{hyperref}
\usepackage{xcolor}
\usepackage{pdfpages}

\makeatletter
\AtBeginDocument{%
  \let\LS@rot\@undefined
}
\makeatother

\newtheorem{theorem}{Theorem}
\newcommand{\inlinesection}[1]{\textit{#1}---}

\begin{document}
\title{Analytic correspondence between multipartite entanglement and quantum phase transitions}
\author{Huynh Le Dan Linh}
\affiliation{University of Information Technology, Ho Chi Minh City, Vietnam}
\affiliation{Vietnam National University, Ho Chi Minh City, Vietnam}
\author{Vu Tuan Hai}
\affiliation{University of Information Technology, Ho Chi Minh City, Vietnam}
\affiliation{Vietnam National University, Ho Chi Minh City, Vietnam}
\author{Le Bin Ho}
\affiliation{Frontier Research Institute for Interdisciplinary Sciences, Tohoku University, Sendai 980-8578, Japan}
\affiliation{Department of Applied Physics, Graduate School of Engineering, Tohoku University, Sendai 980-8579, Japan}
\email{ho.bin.le.e3@tohoku.ac.jp}
\date{\today}

\begin{abstract}
We derive an analytic correspondence between multipartite concentratable entanglement (CE) and quantum phase transitions in one-dimensional quantum spin systems. We prove that CE shares the same analyticity structure as generalized order parameters, identifies the same quantum phases, and, for Gaussian ground states, is completely determined by the same single-particle correlation matrix. Numerical validation on the transverse-field Ising and generalized
cluster-Ising models confirms these analytical predictions for both
symmetry-breaking and symmetry-protected topological quantum phase
transitions. Since CE can be measured directly using a constant-depth parallelized SWAP-test circuit, our results establish CE as an experimentally accessible, model-independent probe of quantum phase transitions without requiring model-specific order parameters.
\end{abstract}
\maketitle

\inlinesection{Introduction}
Quantum phase transitions (QPTs) arise from nonanalytic changes in the many-body ground state driven by quantum fluctuations~\cite{Sachdev2011}. Characterizing quantum phases traditionally relies on phase-specific order parameters whose forms depend on the underlying physical mechanism. While Landau order parameters characterize spontaneous symmetry breaking (SSB), topological and symmetry protected topological (SPT) phases instead require nonlocal string order parameters or topological invariants~\cite{Wen2004,Chen2012,Verresen2017}. Although these descriptions take different forms, they all identify the same critical behavior that separates distinct quantum phases. Yet it remains unclear whether they can be understood within a single, model-independent framework.

Quantum entanglement offers an alternative perspective on many-body
phases and exhibits distinctive signatures near quantum
criticality~\cite{Osterloh2002,Vidal2003,RevModPhys.80.517}. These observations naturally suggest that multipartite entanglement may provide a unified
description of QPTs beyond phase-specific order
parameters. However, despite extensive studies, no analytical framework
establishing such a correspondence has been developed.

Recently, Beckey \emph{et al.}\ introduced concentratable entanglement
(CE), an operational multipartite entanglement measure that can be
efficiently estimated using a constant-depth parallelized SWAP
test~\cite{Beckey2021}. The framework has been extended to mixed
states~\cite{PhysRevA.107.062425} and to qudit and optical
systems~\cite{Foulds2024}, broadening its experimental applicability.
Meanwhile, other multipartite entanglement measures have also been
shown numerically to detect QPTs in specific spin
models~\cite{Li2024}. However, these studies remain empirical and
model-dependent, and a general analytical connection between CE and
quantum criticality remains unknown. In particular, no rigorous
correspondence between CE and conventional order parameters has been
established.

In this Letter, we derive such a correspondence by proving that CE and
generalized order parameters (GOPs) share the same analyticity
structure, distinguish the same quantum phases, and, for Gaussian
ground states, are explicit functions of the same single-particle
correlation matrix. These results provide a rigorous analytical foundation for CE as a model-independent probe of quantum criticality, offering a unified and experimentally accessible alternative to phase-specific order
parameters.

\inlinesection{Hamiltonians and quantum phases}
Consider a family of local spin-$1/2$ Hamiltonians
\(
H(\boldsymbol{\lambda})
=
\sum_{\alpha}
\lambda_{\alpha} h_{\alpha},
\)
acting on
$\mathcal{H}=(\mathbb{C}^{2})^{\otimes N}$
and analytic in the control parameters
$\boldsymbol{\lambda}\in\mathbb{R}^{d}$.
We assume a unique gapped ground state
$\ket{\psi_{0}(\boldsymbol{\lambda})}$
within each quantum phase. The critical manifold
$\mathcal{B}$
separates distinct phases in the thermodynamic limit, where the
assumption of a unique gapped ground state breaks down through gap
closing or ground-state degeneracy~\cite{Sachdev2011,Kato1995}.

For SPT phases, an additional notion
of phase equivalence is required. Two gapped phases are said to be
finite-depth local unitary (FDLU) inequivalent if no finite-depth local
unitary circuit connects their ground states while preserving the
gap~\cite{Chen2010}.

\inlinesection{Generalized order parameters}
We consider GOP of the form
\begin{equation}
\mathcal{O}(\lambda)
=
\frac{1}{N-k+1}
\sum_{i=1}^{N-k+1}
\langle \hat{Q}_{i}\rangle_{\lambda},
\label{eq:O_general}
\end{equation}
where 
$k$ is the length of the Pauli string,
and 
$
\hat{Q}_{i}
=
\bigotimes_{j=0}^{k-1}
P_{i+j},
$ with $P_{i+j}\in\{I,X,Y,Z\}$. This unified form includes conventional local, string, and cluster order parameters used to characterize quantum phases.

\inlinesection{Concentratable entanglement}
For an $N$-qubit pure state $\ket{\psi}$ with qubit set
$S=\{1,\ldots,N\}$, the CE associated
with a subset $s\subseteq S$ is defined as~\cite{Beckey2021}
\begin{equation}
\mathcal{C}(s)
=
1-
\frac{1}{2^{|s|}}
\sum_{\alpha\subseteq s}
\operatorname{Tr}\rho_{\alpha}^{2},
\label{eq:CE_def}
\end{equation}
where $\rho_{\alpha}$ is the reduced density matrix of subsystem
$\alpha$. The full-system CE, obtained by setting $s=S$, satisfies
\begin{equation}
\mathcal{C}(S)
=
1-
\frac{1}{2^{N}}
\sum_{\alpha\subseteq S}
\operatorname{Tr}\rho_{\alpha}^{2}
=
1-p(\mathbf{0}),
\label{eq:CE_full}
\end{equation}
where $p(\mathbf{0})$ is the probability of the all-zero outcome in
the parallelized SWAP test~\cite{Beckey2021}. Unless specified otherwise,
we write $\mathcal C$ for the full-system quantity
$\mathcal C(S)$. We use the following established properties:
$\mathcal{C}(s)\geq0$, with equality only for fully separable states;
$\mathcal{C}(s')\leq\mathcal{C}(s)$ for $s'\subseteq s$; and continuity
under variations of the state~\cite{Beckey2021}.

\begin{theorem}[Shared analyticity]
\label{thm:analyticity}
Let $H(\boldsymbol{\lambda})$ be a real-analytic family of local
Hamiltonians with a unique ground state
$\ket{\psi_0(\boldsymbol{\lambda})}$ separated from the rest of the
spectrum by a finite gap throughout an open connected region
$\mathcal{R}\subset\mathbb{R}^{d}$. Then both the GOP $\mathcal{O}(\boldsymbol{\lambda})$ and the full-system
CE
$\mathcal{C}(S;\boldsymbol{\lambda})$ are real-analytic on
$\mathcal{R}$. Consequently, a nonanalyticity of either quantity can
occur only where the ground state ceases to be isolated, including the
critical manifold $\mathcal{B}$ in the thermodynamic limit.
\end{theorem}

\begin{proof}
Analytic perturbation theory~\cite{Kato1995,Reed1978} implies that the
ground-state projector
\(
P_0(\boldsymbol{\lambda})
=
\ket{\psi_0(\boldsymbol{\lambda})}
\bra{\psi_0(\boldsymbol{\lambda})}
\)
is real-analytic wherever the ground state remains nondegenerate and
spectrally isolated. Writing the GOP as
\(
\mathcal{O}(\boldsymbol{\lambda})
=
\operatorname{Tr}
\!\left[
\hat{\mathcal{O}}\,
P_0(\boldsymbol{\lambda})
\right],
\)
shows that $\mathcal{O}$ is analytic because it is a linear functional
of $P_0$. Likewise, each reduced density matrix
\(
\rho_{\alpha}(\boldsymbol{\lambda})
=
\operatorname{Tr}_{\bar{\alpha}}
P_0(\boldsymbol{\lambda})
\)
is analytic. Its purity
$\operatorname{Tr}\rho_{\alpha}^{2}(\boldsymbol{\lambda})$ is therefore
analytic, and so is the finite sum
\(
\mathcal{C}(S;\boldsymbol{\lambda})
=
1-\frac{1}{2^{N}}
\sum_{\alpha\subseteq S}
\operatorname{Tr}
\rho_{\alpha}^{2}(\boldsymbol{\lambda}).
\)
\end{proof}

Theorem~\ref{thm:analyticity} establishes a common analyticity domain
for $\mathcal{O}$ and $\mathcal{C}$, but does not by itself guarantee
that either quantity becomes singular at every point of
$\mathcal{B}$. We next identify conditions under which CE
distinguishes the phases separated by the transition.

\begin{theorem}[Phase discrimination by CE]
\label{thm:main}
Let $H(\boldsymbol{\lambda})$ possess
$\mathbb{Z}_2\times\mathbb{Z}_2$ symmetry with gapped
SPT and trivial phases separated by
the critical manifold $\mathcal{B}$. Then the full-system CE $\mathcal{C}(S)$ satisfies
\begin{align}
\mathcal C(S)
&\ge\frac14,
&&\text{throughout the SPT phase},
\\
\mathcal C(S)
&\rightarrow0,
&&\text{upon approaching the trivial phase.}
\end{align}
\end{theorem}

\begin{proof}
Within the SPT phase,
$\mathbb{Z}_2\times\mathbb{Z}_2$ symmetry enforces
$\langle\sigma_i^\mu\rangle=0$
for every site and spin component (see Supplemental Material),
implying the maximally mixed one-site reduced state
$\rho_i=\mathbbm{1}_2/2$.
The Bloch decomposition then gives
\(
\mathcal{C}(\{i\})=\frac14.
\)
Monotonicity of CE under subsystem inclusion yields
\(
\mathcal{C}(S)
\ge
\mathcal{C}(\{i\})
=
\frac14.
\)

In the trivial phases, including the paramagnetic (PM) and
antiferromagnetic (AFM) phases, the ground state continuously
approaches a product state, for which
$\mathcal{C}(S)=0$.
Since CE is continuous in the trace norm,
\(
\mathcal{C}(S)\rightarrow0
\)
as the trivial limit is approached.
\end{proof}

Combining Theorems~\ref{thm:analyticity}
and~\ref{thm:main}, CE is analytic throughout each gapped phase while
distinguishing adjacent quantum phases. It therefore provides an
order-parameter-independent indicator of QPTs.

The same single-site argument applies equally to finite-size
representatives of SSB phases. Although such systems do not
realize
$\mathbb Z_2\times\mathbb Z_2$
SPT order, their symmetric finite-size ground states likewise possess
maximally mixed one-site reduced density matrices. Consequently, the
same argument yields the lower bound
$\mathcal{C}(S)\ge1/4$
throughout the symmetry-breaking phase, whereas
$\mathcal{C}(S)\rightarrow0$
in the product-state limit of the corresponding PM phase.

\begin{figure*}[t]
  \centering
  \includegraphics[width=\textwidth]{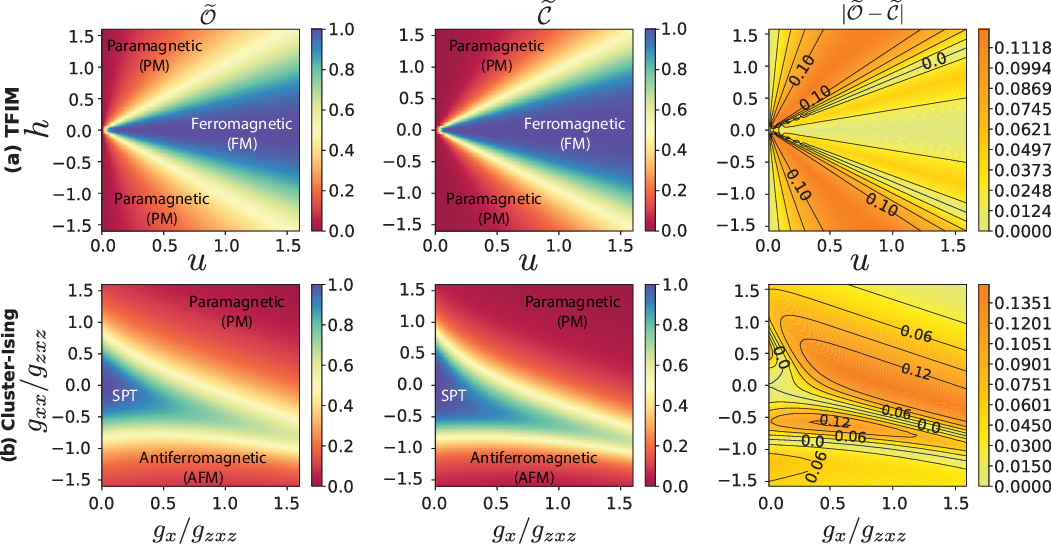}
\caption{
Comparison between the normalized GOP
$\widetilde{\mathcal O}$ (left), normalized CE
$\widetilde{\mathcal C}$ (center), and their absolute difference
$|\widetilde{\mathcal O}-\widetilde{\mathcal C}|$ (right) for two
representative one-dimensional spin models with $N=5$:
(a) the transverse-field Ising model,
and
(b) the generalized cluster-Ising model.
In all cases, CE accurately reproduces the phase boundaries identified
by the corresponding GOP. The normalized
root-mean-square deviations are approximately $7\%$ and
$8\%$, respectively. The remaining discrepancies are localized
predominantly near the critical manifold, while the agreement within
each gapped phase is excellent, consistent with
Theorem~\ref{thm:analyticity}.
}
  \label{fig1}
\end{figure*}

\inlinesection{Gaussian-state correspondence}
For Gaussian ground states of free-fermion Hamiltonians, the correspondence
between CE and GOP
becomes exact. The Gaussian ground state is completely characterized by
the single-particle correlation matrix
\(
\Gamma_{ij}
=
\langle
c_i^\dagger c_j
\rangle,
\)
from which every reduced density matrix is determined through~\cite{Peschel2003,Peschel2009}
\begin{align}
    \operatorname{Tr}\rho_s^2
=
\det\!\left[
\Gamma_s^2+
(\mathbbm{1}-\Gamma_s)^2
\right].
\end{align}
As a result, both the CE and GOP admit the same
correlation matrix, leading to the following theorem.

\begin{theorem}[Exact correspondence for Gaussian ground states]
\label{thm:ff}
For a Gaussian ground state of a free-fermion Hamiltonian, both the
full-system CE $\mathcal{C}(S)$ and every
GOP $\mathcal{O}$ are explicit analytic
functions of the same single-particle correlation matrix $\Gamma$:
\begin{align}
\mathcal{C}(S)
&=
1-
\frac{1}{2^{N}}
\sum_{s\subseteq S}
\det\!\left[
\Gamma_s^2+
(\mathbbm{1}-\Gamma_s)^2
\right],
\label{eq:CE_ff}
\\
\mathcal{O}
&=
f(\Gamma),
\label{eq:O_ff}
\end{align}
where $f$ is an analytic function determined by the corresponding Pauli
string. Consequently, $\mathcal{C}(S)$ and $\mathcal{O}$ become
nonanalytic on the same critical manifold.
\end{theorem}

\begin{proof}
Equation~\eqref{eq:CE_ff} follows directly from the Gaussian purity
formula (see Supplemental Material) together with the definition of
CE in Eq.~\eqref{eq:CE_def}. Under the
Jordan--Wigner transformation, every Pauli string is mapped to a
fermionic operator. By Wick's theorem, its expectation value in a
Gaussian ground state is an analytic function of the correlation matrix
$\Gamma$~\cite{Bravyi2005}. Hence every GOP can
be written as the analytic function $f(\Gamma)$. Since $\Gamma$ is
analytic whenever the single-particle gap remains open, both
$\mathcal{C}(S)$ and $\mathcal{O}$ are analytic throughout each gapped
phase. The single-particle and many-body gaps close on the same
critical manifold, implying that both quantities become nonanalytic at
the same transition.
\end{proof}

\inlinesection{Numerical validation}
We validate the theoretical predictions on two representative
one-dimensional spin models spanning SSB and
SPT phases: the
transverse-field Ising model (TFIM) and the generalized cluster-Ising model. In each case, the
full-system CE is evaluated as
$\mathcal C(S)=1-p(\mathbf0)$ using the parallelized SWAP-test
expression of Eq.~\eqref{eq:CE_full}. Figure~\ref{fig1} compares the
normalized GOP
$\widetilde{\mathcal O}$,
the normalized CE
$\widetilde{\mathcal C}$,
and their absolute difference over the corresponding phase diagrams.

Across all these models, CE faithfully reproduces the phase boundaries
identified by the GOP. The agreement is
essentially exact for free-fermion systems, where deviations are concentrated primarily near
the critical manifold rather than within the bulk phases.

\inlinesection{Transverse-field Ising model}
The transverse-field Ising model (TFIM) is the paradigmatic exactly
solvable model of one-dimensional quantum criticality and the canonical
example of an SSB transition~\cite{PhysRevLett.25.443,Sachdev2011}.
Its Hamiltonian is
\begin{equation}
H_{\mathrm{TFIM}}
=
-u\sum_{j=1}^{N-1}Z_jZ_{j+1}
-
h\sum_{j=1}^{N}X_j,
\end{equation}
which is invariant under the global $\mathbb Z_2$ spin-flip symmetry.
In the thermodynamic limit, the competition between the interaction
and transverse field gives rise to a ferromagnetic (FM) phase
($|u|>|h|$) where the $\mathbb Z_2$ symmetry is spontaneously broken,
and a PM phase ($|u|<|h|$) where the symmetry remains
preserved. The corresponding GOP is
\[
\mathcal O
=
\frac{1}{N-1}
\sum_{i=1}^{N-1}
\langle Z_iZ_{i+1}\rangle .
\]

Since the model is exactly solvable by free fermions, Theorem~\ref{thm:ff}
applies directly. Exact diagonalization for $N=5$ on a $64\times64$
parameter grid over $h\in[-1.6,1.6]$ and $u\in[0.0,1.6]$ yields a normalized RMS
deviation of approximately $7\%$. 
The critical lines
$|u|=|h|$ 
separating the FM and PM
phases are reproduced almost identically by
$\widetilde{\mathcal O}$
and
$\widetilde{\mathcal C}$, with noticeable discrepancies appearing only
in the immediate vicinity of the transition
[Fig.~\ref{fig1}(a)].

\inlinesection{Generalized cluster-Ising model}
The generalized cluster-Ising model provides a paradigmatic example of
an SPT phase, where the ground state is
characterized by nonlocal string order rather than SSB~\cite{Verresen2017,Smacchia2011, Cong2019}. Its Hamiltonian is
\begin{align}
H
=
&-g_{zxz}
\sum_{i=1}^{N-2}
Z_iX_{i+1}Z_{i+2}\nonumber\\
&-
g_x
\sum_{i=1}^{N}
X_i
-
g_{xx}
\sum_{i=1}^{N-1}
X_iX_{i+1},
\end{align}
with the GOP is 
\[
\mathcal O
=
\frac{1}{N-2}
\sum_{i=1}^{N-2}
\langle
Z_iX_{i+1}Z_{i+2}
\rangle,
\]
that distinguishes the SPT phase from the surrounding trivial PM and AFM phases~\cite{Cong2019}.

Along the exactly solvable line
$g_{xx}=0$,
Theorem~\ref{thm:ff} predicts an exact correspondence through the
Gaussian correlation matrix.
Exact diagonalization for $N=5$
gives a normalized RMS deviation of approximately $8\%$ over the
$(g_x,g_{xx})$ parameter plane.
The boundary separating the SPT phase from the trivial PM and AFM phases is accurately reproduced by CE, again with discrepancies localized near the critical manifold
[Fig.~\ref{fig1}(b)].

The numerical results consistently support the theoretical framework.
The agreement is strongest for exactly solvable free-fermion models,
where Theorem~\ref{thm:ff} predicts an exact correspondence.
Across all these examples, the differences remain concentrated near the
critical manifold, confirming that the CE captures
the same quantum-critical structure as the corresponding GOP.

\inlinesection{Conclusion}
We have established an analytic correspondence between multipartite
entanglement and QPTs in one-dimensional quantum
spin systems. Specifically, we have shown that CE shares the same analyticity structure as GOP, distinguishes the same quantum phases, and, for Gaussian
ground states, becomes an explicit function of the same single-particle
correlation matrix. Numerical results for representative symmetry-breaking and
symmetry-protected topological models further support these analytical
predictions.

Beyond the specific models considered here,
our results demonstrate that CE
can serve not only as a multipartite entanglement measure but also as a
unified probe of quantum criticality. Unlike conventional approaches,
which require constructing model-dependent order parameters, the same
experimentally measurable quantity applies across distinct classes of
QPTs. Since CE can be measured directly using the
parallelized SWAP test~\cite{Beckey2021}, the present framework opens a
practical route toward experimentally identifying QPTs on programmable quantum devices.

Natural extensions include mixed states, finite-temperature QPTs, higher-dimensional systems, and topological phases
characterized by nonlocal invariants, and experimental implementation
on superconducting and trapped-ion quantum processors.

\inlinesection{Acknowledgments} This work is supported by the Tohoku Initiative for Fostering Global Researchers for Interdisciplinary Sciences (TI-FRIS) of MEXT's Strategic Professional Development Program for Young Researchers. V.T.H.\ is partially supported by the VNUHCM--University of Information Technology's Scientific Research Support Fund.

\inlinesection{Data availability}
There are no publicly available
research data or software supporting this manuscript.
Requests for further information or data should be sent
to the authors.

\bibliography{refs}

\begin{thebibliography}{20}%
\makeatletter
\providecommand \@ifxundefined [1]{%
 \@ifx{#1\undefined}
}%
\providecommand \@ifnum [1]{%
 \ifnum #1\expandafter \@firstoftwo
 \else \expandafter \@secondoftwo
 \fi
}%
\providecommand \@ifx [1]{%
 \ifx #1\expandafter \@firstoftwo
 \else \expandafter \@secondoftwo
 \fi
}%
\providecommand \natexlab [1]{#1}%
\providecommand \enquote  [1]{``#1''}%
\providecommand \bibnamefont  [1]{#1}%
\providecommand \bibfnamefont [1]{#1}%
\providecommand \citenamefont [1]{#1}%
\providecommand \href@noop [0]{\@secondoftwo}%
\providecommand \href [0]{\begingroup \@sanitize@url \@href}%
\providecommand \@href[1]{\@@startlink{#1}\@@href}%
\providecommand \@@href[1]{\endgroup#1\@@endlink}%
\providecommand \@sanitize@url [0]{\catcode `\\12\catcode `\$12\catcode
  `\&12\catcode `\#12\catcode `\^12\catcode `\_12\catcode `\%12\relax}%
\providecommand \@@startlink[1]{}%
\providecommand \@@endlink[0]{}%
\providecommand \url  [0]{\begingroup\@sanitize@url \@url }%
\providecommand \@url [1]{\endgroup\@href {#1}{\urlprefix }}%
\providecommand \urlprefix  [0]{URL }%
\providecommand \Eprint [0]{\href }%
\providecommand \doibase [0]{https://doi.org/}%
\providecommand \selectlanguage [0]{\@gobble}%
\providecommand \bibinfo  [0]{\@secondoftwo}%
\providecommand \bibfield  [0]{\@secondoftwo}%
\providecommand \translation [1]{[#1]}%
\providecommand \BibitemOpen [0]{}%
\providecommand \bibitemStop [0]{}%
\providecommand \bibitemNoStop [0]{.\EOS\space}%
\providecommand \EOS [0]{\spacefactor3000\relax}%
\providecommand \BibitemShut  [1]{\csname bibitem#1\endcsname}%
\let\auto@bib@innerbib\@empty
\bibitem [{\citenamefont {Sachdev}(2011)}]{Sachdev2011}%
  \BibitemOpen
  \bibfield  {author} {\bibinfo {author} {\bibfnamefont {S.}~\bibnamefont
  {Sachdev}},\ }\href {https://doi.org/10.1017/CBO9780511973765} {\emph
  {\bibinfo {title} {Quantum Phase Transitions}}},\ \bibinfo {edition} {2nd}\
  ed.\ (\bibinfo  {publisher} {Cambridge University Press},\ \bibinfo {address}
  {Cambridge, UK},\ \bibinfo {year} {2011})\BibitemShut {NoStop}%
\bibitem [{\citenamefont {Wen}(2004)}]{Wen2004}%
  \BibitemOpen
  \bibfield  {author} {\bibinfo {author} {\bibfnamefont {X.-G.}\ \bibnamefont
  {Wen}},\ }\href@noop {} {\emph {\bibinfo {title} {Quantum Field Theory of
  Many-Body Systems: From the Origin of Sound to an Origin of Light and
  Electrons}}},\ Oxford Graduate Texts\ (\bibinfo  {publisher} {Oxford
  University Press},\ \bibinfo {address} {Oxford},\ \bibinfo {year}
  {2004})\BibitemShut {NoStop}%
\bibitem [{\citenamefont {Chen}\ \emph {et~al.}(2013)\citenamefont {Chen},
  \citenamefont {Gu}, \citenamefont {Liu},\ and\ \citenamefont
  {Wen}}]{Chen2012}%
  \BibitemOpen
  \bibfield  {author} {\bibinfo {author} {\bibfnamefont {X.}~\bibnamefont
  {Chen}}, \bibinfo {author} {\bibfnamefont {Z.-C.}\ \bibnamefont {Gu}},
  \bibinfo {author} {\bibfnamefont {Z.-X.}\ \bibnamefont {Liu}},\ and\ \bibinfo
  {author} {\bibfnamefont {X.-G.}\ \bibnamefont {Wen}},\ }\bibfield  {title}
  {\bibinfo {title} {Symmetry protected topological orders and the group
  cohomology of their symmetry group},\ }\href
  {https://doi.org/10.1103/PhysRevB.87.155114} {\bibfield  {journal} {\bibinfo
  {journal} {Phys. Rev. B}\ }\textbf {\bibinfo {volume} {87}},\ \bibinfo
  {pages} {155114} (\bibinfo {year} {2013})}\BibitemShut {NoStop}%
\bibitem [{\citenamefont {Verresen}\ \emph {et~al.}(2017)\citenamefont
  {Verresen}, \citenamefont {Moessner},\ and\ \citenamefont
  {Pollmann}}]{Verresen2017}%
  \BibitemOpen
  \bibfield  {author} {\bibinfo {author} {\bibfnamefont {R.}~\bibnamefont
  {Verresen}}, \bibinfo {author} {\bibfnamefont {R.}~\bibnamefont {Moessner}},\
  and\ \bibinfo {author} {\bibfnamefont {F.}~\bibnamefont {Pollmann}},\
  }\bibfield  {title} {\bibinfo {title} {One-dimensional symmetry protected
  topological phases and their transitions},\ }\href
  {https://doi.org/10.1103/PhysRevB.96.165124} {\bibfield  {journal} {\bibinfo
  {journal} {Phys. Rev. B}\ }\textbf {\bibinfo {volume} {96}},\ \bibinfo
  {pages} {165124} (\bibinfo {year} {2017})}\BibitemShut {NoStop}%
\bibitem [{\citenamefont {Osterloh}\ \emph {et~al.}(2002)\citenamefont
  {Osterloh}, \citenamefont {Amico}, \citenamefont {Falci},\ and\ \citenamefont
  {Fazio}}]{Osterloh2002}%
  \BibitemOpen
  \bibfield  {author} {\bibinfo {author} {\bibfnamefont {A.}~\bibnamefont
  {Osterloh}}, \bibinfo {author} {\bibfnamefont {L.}~\bibnamefont {Amico}},
  \bibinfo {author} {\bibfnamefont {G.}~\bibnamefont {Falci}},\ and\ \bibinfo
  {author} {\bibfnamefont {R.}~\bibnamefont {Fazio}},\ }\bibfield  {title}
  {\bibinfo {title} {Scaling of entanglement close to a quantum phase
  transition},\ }\href {https://doi.org/10.1038/416608a} {\bibfield  {journal}
  {\bibinfo  {journal} {Nature}\ }\textbf {\bibinfo {volume} {416}},\ \bibinfo
  {pages} {608} (\bibinfo {year} {2002})}\BibitemShut {NoStop}%
\bibitem [{\citenamefont {Vidal}\ \emph {et~al.}(2003)\citenamefont {Vidal},
  \citenamefont {Latorre}, \citenamefont {Rico},\ and\ \citenamefont
  {Kitaev}}]{Vidal2003}%
  \BibitemOpen
  \bibfield  {author} {\bibinfo {author} {\bibfnamefont {G.}~\bibnamefont
  {Vidal}}, \bibinfo {author} {\bibfnamefont {J.~I.}\ \bibnamefont {Latorre}},
  \bibinfo {author} {\bibfnamefont {E.}~\bibnamefont {Rico}},\ and\ \bibinfo
  {author} {\bibfnamefont {A.}~\bibnamefont {Kitaev}},\ }\bibfield  {title}
  {\bibinfo {title} {Entanglement in quantum critical phenomena},\ }\href
  {https://doi.org/10.1103/PhysRevLett.90.227902} {\bibfield  {journal}
  {\bibinfo  {journal} {Phys. Rev. Lett.}\ }\textbf {\bibinfo {volume} {90}},\
  \bibinfo {pages} {227902} (\bibinfo {year} {2003})}\BibitemShut {NoStop}%
\bibitem [{\citenamefont {Amico}\ \emph {et~al.}(2008)\citenamefont {Amico},
  \citenamefont {Fazio}, \citenamefont {Osterloh},\ and\ \citenamefont
  {Vedral}}]{RevModPhys.80.517}%
  \BibitemOpen
  \bibfield  {author} {\bibinfo {author} {\bibfnamefont {L.}~\bibnamefont
  {Amico}}, \bibinfo {author} {\bibfnamefont {R.}~\bibnamefont {Fazio}},
  \bibinfo {author} {\bibfnamefont {A.}~\bibnamefont {Osterloh}},\ and\
  \bibinfo {author} {\bibfnamefont {V.}~\bibnamefont {Vedral}},\ }\bibfield
  {title} {\bibinfo {title} {Entanglement in many-body systems},\ }\href
  {https://doi.org/10.1103/RevModPhys.80.517} {\bibfield  {journal} {\bibinfo
  {journal} {Rev. Mod. Phys.}\ }\textbf {\bibinfo {volume} {80}},\ \bibinfo
  {pages} {517} (\bibinfo {year} {2008})}\BibitemShut {NoStop}%
\bibitem [{\citenamefont {Beckey}\ \emph {et~al.}(2021)\citenamefont {Beckey},
  \citenamefont {Gigena}, \citenamefont {Coles},\ and\ \citenamefont
  {Cerezo}}]{Beckey2021}%
  \BibitemOpen
  \bibfield  {author} {\bibinfo {author} {\bibfnamefont {J.~L.}\ \bibnamefont
  {Beckey}}, \bibinfo {author} {\bibfnamefont {N.}~\bibnamefont {Gigena}},
  \bibinfo {author} {\bibfnamefont {P.~J.}\ \bibnamefont {Coles}},\ and\
  \bibinfo {author} {\bibfnamefont {M.}~\bibnamefont {Cerezo}},\ }\bibfield
  {title} {\bibinfo {title} {Computable and operationally meaningful
  multipartite entanglement measures},\ }\href
  {https://doi.org/10.1103/PhysRevLett.127.140501} {\bibfield  {journal}
  {\bibinfo  {journal} {Phys. Rev. Lett.}\ }\textbf {\bibinfo {volume} {127}},\
  \bibinfo {pages} {140501} (\bibinfo {year} {2021})}\BibitemShut {NoStop}%
\bibitem [{\citenamefont {Beckey}\ \emph {et~al.}(2023)\citenamefont {Beckey},
  \citenamefont {Pelegr\'{\i}}, \citenamefont {Foulds},\ and\ \citenamefont
  {Pearson}}]{PhysRevA.107.062425}%
  \BibitemOpen
  \bibfield  {author} {\bibinfo {author} {\bibfnamefont {J.~L.}\ \bibnamefont
  {Beckey}}, \bibinfo {author} {\bibfnamefont {G.}~\bibnamefont
  {Pelegr\'{\i}}}, \bibinfo {author} {\bibfnamefont {S.}~\bibnamefont
  {Foulds}},\ and\ \bibinfo {author} {\bibfnamefont {N.~J.}\ \bibnamefont
  {Pearson}},\ }\bibfield  {title} {\bibinfo {title} {Multipartite entanglement
  measures via bell-basis measurements},\ }\href
  {https://doi.org/10.1103/PhysRevA.107.062425} {\bibfield  {journal} {\bibinfo
   {journal} {Phys. Rev. A}\ }\textbf {\bibinfo {volume} {107}},\ \bibinfo
  {pages} {062425} (\bibinfo {year} {2023})}\BibitemShut {NoStop}%
\bibitem [{\citenamefont {Foulds}\ \emph {et~al.}(2024)\citenamefont {Foulds},
  \citenamefont {Prove},\ and\ \citenamefont {Kendon}}]{Foulds2024}%
  \BibitemOpen
  \bibfield  {author} {\bibinfo {author} {\bibfnamefont {S.}~\bibnamefont
  {Foulds}}, \bibinfo {author} {\bibfnamefont {O.}~\bibnamefont {Prove}},\ and\
  \bibinfo {author} {\bibfnamefont {V.}~\bibnamefont {Kendon}},\ }\bibfield
  {title} {\bibinfo {title} {Generalizing multipartite concentratable
  entanglement for practical applications: mixed, qudit and optical states},\
  }\href {https://doi.org/10.1098/rsta.2024.0411} {\bibfield  {journal}
  {\bibinfo  {journal} {Philosophical Transactions of the Royal Society A:
  Mathematical, Physical and Engineering Sciences}\ }\textbf {\bibinfo {volume}
  {382}},\ \bibinfo {pages} {20240411} (\bibinfo {year} {2024})}\BibitemShut
  {NoStop}%
\bibitem [{\citenamefont {Li}\ \emph {et~al.}(2024)\citenamefont {Li},
  \citenamefont {Zhou}, \citenamefont {Zhang}, \citenamefont {Bai},\ and\
  \citenamefont {Lin}}]{Li2024}%
  \BibitemOpen
  \bibfield  {author} {\bibinfo {author} {\bibfnamefont {Y.-C.}\ \bibnamefont
  {Li}}, \bibinfo {author} {\bibfnamefont {Y.-H.}\ \bibnamefont {Zhou}},
  \bibinfo {author} {\bibfnamefont {Y.}~\bibnamefont {Zhang}}, \bibinfo
  {author} {\bibfnamefont {Y.-K.}\ \bibnamefont {Bai}},\ and\ \bibinfo {author}
  {\bibfnamefont {H.-Q.}\ \bibnamefont {Lin}},\ }\bibfield  {title} {\bibinfo
  {title} {Multipartite entanglement serves as a faithful detector for quantum
  phase transitions},\ }\href {https://doi.org/10.1088/1367-2630/ad273a}
  {\bibfield  {journal} {\bibinfo  {journal} {New Journal of Physics}\ }\textbf
  {\bibinfo {volume} {26}},\ \bibinfo {pages} {023031} (\bibinfo {year}
  {2024})}\BibitemShut {NoStop}%
\bibitem [{\citenamefont {Kato}(1995)}]{Kato1995}%
  \BibitemOpen
  \bibfield  {author} {\bibinfo {author} {\bibfnamefont {T.}~\bibnamefont
  {Kato}},\ }\href {https://doi.org/10.1007/978-3-642-66282-9} {\emph {\bibinfo
  {title} {Perturbation Theory for Linear Operators}}},\ \bibinfo {edition}
  {2nd}\ ed.,\ Classics in Mathematics\ (\bibinfo  {publisher} {Springer},\
  \bibinfo {address} {Berlin, Heidelberg},\ \bibinfo {year} {1995})\BibitemShut
  {NoStop}%
\bibitem [{\citenamefont {Chen}\ \emph {et~al.}(2010)\citenamefont {Chen},
  \citenamefont {Gu},\ and\ \citenamefont {Wen}}]{Chen2010}%
  \BibitemOpen
  \bibfield  {author} {\bibinfo {author} {\bibfnamefont {X.}~\bibnamefont
  {Chen}}, \bibinfo {author} {\bibfnamefont {Z.-C.}\ \bibnamefont {Gu}},\ and\
  \bibinfo {author} {\bibfnamefont {X.-G.}\ \bibnamefont {Wen}},\ }\bibfield
  {title} {\bibinfo {title} {Local unitary transformation, long-range quantum
  entanglement, wave function renormalization, and topological order},\ }\href
  {https://doi.org/10.1103/PhysRevB.82.155138} {\bibfield  {journal} {\bibinfo
  {journal} {Phys. Rev. B}\ }\textbf {\bibinfo {volume} {82}},\ \bibinfo
  {pages} {155138} (\bibinfo {year} {2010})}\BibitemShut {NoStop}%
\bibitem [{\citenamefont {Reed}\ and\ \citenamefont {Simon}(1978)}]{Reed1978}%
  \BibitemOpen
  \bibfield  {author} {\bibinfo {author} {\bibfnamefont {M.}~\bibnamefont
  {Reed}}\ and\ \bibinfo {author} {\bibfnamefont {B.}~\bibnamefont {Simon}},\
  }\href@noop {} {\emph {\bibinfo {title} {Methods of Modern Mathematical
  Physics. Vol. IV: Analysis of Operators}}},\ \bibinfo {series} {Methods of
  Modern Mathematical Physics}, Vol.~\bibinfo {volume} {4}\ (\bibinfo
  {publisher} {Academic Press},\ \bibinfo {address} {New York},\ \bibinfo
  {year} {1978})\BibitemShut {NoStop}%
\bibitem [{\citenamefont {Peschel}(2003)}]{Peschel2003}%
  \BibitemOpen
  \bibfield  {author} {\bibinfo {author} {\bibfnamefont {I.}~\bibnamefont
  {Peschel}},\ }\bibfield  {title} {\bibinfo {title} {Calculation of reduced
  density matrices from correlation functions},\ }\href
  {https://doi.org/10.1088/0305-4470/36/14/101} {\bibfield  {journal} {\bibinfo
   {journal} {Journal of Physics A: Mathematical and General}\ }\textbf
  {\bibinfo {volume} {36}},\ \bibinfo {pages} {L205} (\bibinfo {year}
  {2003})}\BibitemShut {NoStop}%
\bibitem [{\citenamefont {Peschel}\ and\ \citenamefont
  {Eisler}(2009)}]{Peschel2009}%
  \BibitemOpen
  \bibfield  {author} {\bibinfo {author} {\bibfnamefont {I.}~\bibnamefont
  {Peschel}}\ and\ \bibinfo {author} {\bibfnamefont {V.}~\bibnamefont
  {Eisler}},\ }\bibfield  {title} {\bibinfo {title} {Reduced density matrices
  and entanglement entropy in free lattice models},\ }\href
  {https://doi.org/10.1088/1751-8113/42/50/504003} {\bibfield  {journal}
  {\bibinfo  {journal} {Journal of Physics A: Mathematical and Theoretical}\
  }\textbf {\bibinfo {volume} {42}},\ \bibinfo {pages} {504003} (\bibinfo
  {year} {2009})}\BibitemShut {NoStop}%
\bibitem [{\citenamefont {Bravyi}(2005)}]{Bravyi2005}%
  \BibitemOpen
  \bibfield  {author} {\bibinfo {author} {\bibfnamefont {S.}~\bibnamefont
  {Bravyi}},\ }\bibfield  {title} {\bibinfo {title} {Lagrangian representation
  for fermionic linear optics},\ }\href@noop {} {\bibfield  {journal} {\bibinfo
   {journal} {Quantum Info. Comput.}\ }\textbf {\bibinfo {volume} {5}},\
  \bibinfo {pages} {216–238} (\bibinfo {year} {2005})}\BibitemShut {NoStop}%
\bibitem [{\citenamefont {Elliott}\ \emph {et~al.}(1970)\citenamefont
  {Elliott}, \citenamefont {Pfeuty},\ and\ \citenamefont
  {Wood}}]{PhysRevLett.25.443}%
  \BibitemOpen
  \bibfield  {author} {\bibinfo {author} {\bibfnamefont {R.~J.}\ \bibnamefont
  {Elliott}}, \bibinfo {author} {\bibfnamefont {P.}~\bibnamefont {Pfeuty}},\
  and\ \bibinfo {author} {\bibfnamefont {C.}~\bibnamefont {Wood}},\ }\bibfield
  {title} {\bibinfo {title} {Ising model with a transverse field},\ }\href
  {https://doi.org/10.1103/PhysRevLett.25.443} {\bibfield  {journal} {\bibinfo
  {journal} {Phys. Rev. Lett.}\ }\textbf {\bibinfo {volume} {25}},\ \bibinfo
  {pages} {443} (\bibinfo {year} {1970})}\BibitemShut {NoStop}%
\bibitem [{\citenamefont {Smacchia}\ \emph {et~al.}(2011)\citenamefont
  {Smacchia}, \citenamefont {Amico}, \citenamefont {Facchi}, \citenamefont
  {Fazio}, \citenamefont {Florio}, \citenamefont {Pascazio},\ and\
  \citenamefont {Vedral}}]{Smacchia2011}%
  \BibitemOpen
  \bibfield  {author} {\bibinfo {author} {\bibfnamefont {P.}~\bibnamefont
  {Smacchia}}, \bibinfo {author} {\bibfnamefont {L.}~\bibnamefont {Amico}},
  \bibinfo {author} {\bibfnamefont {P.}~\bibnamefont {Facchi}}, \bibinfo
  {author} {\bibfnamefont {R.}~\bibnamefont {Fazio}}, \bibinfo {author}
  {\bibfnamefont {G.}~\bibnamefont {Florio}}, \bibinfo {author} {\bibfnamefont
  {S.}~\bibnamefont {Pascazio}},\ and\ \bibinfo {author} {\bibfnamefont
  {V.}~\bibnamefont {Vedral}},\ }\bibfield  {title} {\bibinfo {title}
  {Statistical mechanics of the cluster ising model},\ }\href
  {https://doi.org/10.1103/PhysRevA.84.022304} {\bibfield  {journal} {\bibinfo
  {journal} {Phys.\ Rev.\ A}\ }\textbf {\bibinfo {volume} {84}},\ \bibinfo
  {pages} {022304} (\bibinfo {year} {2011})}\BibitemShut {NoStop}%
\bibitem [{\citenamefont {Cong}\ \emph {et~al.}(2019)\citenamefont {Cong},
  \citenamefont {Choi},\ and\ \citenamefont {Lukin}}]{Cong2019}%
  \BibitemOpen
  \bibfield  {author} {\bibinfo {author} {\bibfnamefont {I.}~\bibnamefont
  {Cong}}, \bibinfo {author} {\bibfnamefont {S.}~\bibnamefont {Choi}},\ and\
  \bibinfo {author} {\bibfnamefont {M.~D.}\ \bibnamefont {Lukin}},\ }\bibfield
  {title} {\bibinfo {title} {Quantum convolutional neural networks},\ }\href
  {https://doi.org/10.1038/s41567-019-0648-8} {\bibfield  {journal} {\bibinfo
  {journal} {Nature Physics}\ }\textbf {\bibinfo {volume} {15}},\ \bibinfo
  {pages} {1273} (\bibinfo {year} {2019})}\BibitemShut {NoStop}%
\end{thebibliography}%

\clearpage
\onecolumngrid
\clearpage



\section*{Supplementary material}

# Supplemental Material for
# "Analytic correspondence between multipartite entanglement and quantum phase transitions"

Huynh Le Dan Linh,[1,2] Vu Tuan Hai,[1,2] and Le Bin Ho[3,4]

[1] *University of Information Technology, Ho Chi Minh City, Vietnam*
[2] *Vietnam National University, Ho Chi Minh City, Vietnam*
[3] *Frontier Research Institute for Interdisciplinary Sciences, Tohoku University, Sendai 980-8578, Japan*
[4] *Department of Applied Physics, Graduate School of Engineering, Tohoku University, Sendai 980-8579, Japan*[*]

## S1. OVERVIEW

This Supplemental Material provides detailed proofs of the three theorems presented in the main text together with a brief description of the numerical implementation. Throughout, we adopt the notation of the main text.

The remainder of this Supplemental Material is organized as follows. Section S2 proves Theorem 1 on the common analyticity of concentratable entanglement (CE) and generalized order parameters (GOP). Section S3 proves Theorem 2, establishing that CE distinguishes symmetry-protected topological (SPT) and trivial phases. Section S4 proves Theorem 3, establishing the exact correspondence for Gaussian ground states. Finally, Section S5 summarizes the numerical implementation used to generate the results presented in the main text.

## S2. PROOF OF THEOREM 1

*Proof.* Let $\mathcal{R} \subset \mathbb{R}^d$ be an open connected region in parameter space where the ground state is unique and separated from the excited spectrum by a finite many-body gap. Denote the corresponding ground-state projector by

$$P_0(\boldsymbol{\lambda}) = |\psi_0(\boldsymbol{\lambda})\rangle\langle\psi_0(\boldsymbol{\lambda})|. \tag{S1}$$

By analytic perturbation theory (Kato–Rellich theorem) [1, 2], $P_0(\boldsymbol{\lambda})$ is a real-analytic function of $\boldsymbol{\lambda}$ throughout $\mathcal{R}$.

The GOP can be written as

$$\mathcal{O}(\boldsymbol{\lambda}) = \mathrm{Tr}\left[\hat{\mathcal{O}} P_0(\boldsymbol{\lambda})\right], \tag{S2}$$

where

$$\hat{\mathcal{O}} = \frac{1}{N-k+1} \sum_{i=1}^{N-k+1} \hat{Q}_i. \tag{S3}$$

Since the trace is a continuous linear functional, $\mathcal{O}(\boldsymbol{\lambda})$ is real analytic on $\mathcal{R}$.

Next consider the full-system CE. For every subsystem $\alpha \subseteq S$,

$$\rho_\alpha(\boldsymbol{\lambda}) = \mathrm{Tr}_{\bar{\alpha}}\, P_0(\boldsymbol{\lambda}) \tag{S4}$$

is obtained from the ground-state projector by a partial trace, which is also a continuous linear map. Hence, $\rho_\alpha(\boldsymbol{\lambda})$ is real analytic throughout $\mathcal{R}$. Its purity, $\mathrm{Tr}\,\rho_\alpha^2$, is a polynomial in the matrix elements of $\rho_\alpha$ and therefore remains real analytic.

Finally,

$$\mathcal{C}(S) = 1 - \frac{1}{2^N} \sum_{\alpha \subseteq S} \mathrm{Tr}\,\rho_\alpha^2 \tag{S5}$$

is a finite sum of real-analytic functions and is therefore itself real analytic on $\mathcal{R}$.

Consequently, both $\mathcal{O}$ and $\mathcal{C}(S)$ are real analytic throughout every gapped phase. Possible nonanalyticities therefore arise only when the ground-state projector ceases to be analytic, namely at points where the many-body gap closes.

□

## S3. PROOF OF THEOREM 2

*Proof.* We first consider the $\mathbb{Z}_2 \times \mathbb{Z}_2$ SPT phase. The protecting symmetry forces all one-site spin expectation values to vanish,

$$\langle \sigma_i^x \rangle = \langle \sigma_i^y \rangle = \langle \sigma_i^z \rangle = 0, \qquad i = 1, \ldots, N, \tag{S6}$$

so that the one-site reduced density matrix takes the Bloch form

$$\rho_i = \frac{1}{2}\left(\mathbb{1} + \mathbf{r}_i \cdot \boldsymbol{\sigma}\right), \tag{S7}$$

with $\mathbf{r}_i = \mathbf{0}$. Hence,

$$\rho_i = \frac{\mathbb{1}_2}{2}, \qquad \mathrm{Tr}\,\rho_i^2 = \frac{1}{2}. \tag{S8}$$

Using the definition of CE (Eq. 2 in the main text),

$$\mathcal{C}(\{i\}) = 1 - \frac{1}{2}\left(1 + \frac{1}{2}\right) = \frac{1}{4}. \tag{S9}$$

[*] ho.bin.le.e3@tohoku.ac.jp

Since CE is monotonic under subsystem inclusion [3],

$$\mathcal{C}(S) \geq \mathcal{C}(\{i\}) = \frac{1}{4}, \tag{S10}$$

which establishes the lower bound throughout the SPT phase.

We next consider the trivial phase. Its ground state is continuously connected to a product state without closing the many-body gap. For a product state,

$$\operatorname{Tr} \rho_\alpha^2 = 1, \qquad \forall\, \alpha \subseteq S, \tag{S11}$$

so that $\mathcal{C}(S) = 0$.

Since CE is continuous with respect to the trace norm,

$$\lim_{\lambda \to \lambda_{\mathrm{triv}}} \mathcal{C}(S) = 0. \tag{S12}$$

Consequently, $\mathcal{C}(S) \geq \frac{1}{4}$ throughout the SPT phase, whereas $\mathcal{C}(S) \to 0$ upon approaching the trivial phase. Consequently, $\mathcal{C}(S)$ distinguishes the two phases. □

## S4. PROOF OF THEOREM 3

*Proof.* We consider a quadratic fermionic Hamiltonian,

$$H = \sum_{i,j=1}^{N} c_i^\dagger A_{ij} c_j + \frac{1}{2} \sum_{i,j=1}^{N} \left( c_i^\dagger B_{ij} c_j^\dagger + \mathrm{H.c.} \right), \tag{S13}$$

where $A = A^\dagger$ and $B = -B^\mathsf{T}$. Such Hamiltonians admit a Bogoliubov transformation that diagonalizes the Hamiltonian and yields a Gaussian ground state completely characterized by the single-particle correlation matrix

$$\Gamma_{ij} = \langle c_i^\dagger c_j \rangle. \tag{S14}$$

For every subsystem $s \subseteq S$, let $\Gamma_s$ denote the restriction of $\Gamma$ to that subsystem. The reduced density matrix remains Gaussian and is uniquely determined by $\Gamma_s$ [4, 5], with purity

$$\operatorname{Tr} \rho_s^2 = \det \left[ \Gamma_s^2 + (\mathbb{1} - \Gamma_s)^2 \right]. \tag{S15}$$

Substituting Eq. (S15) into the CE yields

$$\mathcal{C}(S) = 1 - \frac{1}{2^N} \sum_{s \subseteq S} \det \left[ \Gamma_s^2 + (\mathbb{1} - \Gamma_s)^2 \right], \tag{S16}$$

which is Eq. (7) of the main text. This equation shows that the CE depends solely on the single-particle correlation matrix $\Gamma$.

Next consider the GOP. Under the Jordan–Wigner transformation,

$$\sigma_i^+ = c_i^\dagger \exp \Big( i\pi \sum_{j<i} c_j^\dagger c_j \Big), \tag{S17}$$

every Pauli-string operator is mapped onto a fermionic string operator. Since the ground state is Gaussian, Wick's theorem expresses every fermionic correlation function as a polynomial of two-point correlators [4, 5]. Consequently,

$$\mathcal{O} = f(\Gamma), \tag{S18}$$

where the analytic function $f$ depends on the corresponding Pauli string.

Equations (S16) and (S18) show that both $\mathcal{C}(S)$ and $\mathcal{O}$ are explicit functions of the same correlation matrix $\Gamma$.

Whenever the single-particle gap remains finite, the Bogoliubov transformation depends analytically on the Hamiltonian parameters, implying that $\Gamma$ is real analytic. Therefore, both $\mathcal{C}(S)$ and $\mathcal{O}$ remain analytic throughout every gapped phase. At a quantum phase transition, the closing of the single-particle gap renders $\Gamma$ nonanalytic, and the same singularity is inherited simultaneously by both $\mathcal{C}(S)$ and $\mathcal{O}$. This establishes Theorem 3. □

## S5. NUMERICAL IMPLEMENTATION

This section summarizes the numerical procedures used to generate all results reported in the main text. All numerical results were obtained by exact diagonalization of finite spin chains with open boundary conditions. For each value of the Hamiltonian parameters, the many-body ground state $|\psi_0\rangle$ was computed using double-precision arithmetic.

The GOP was evaluated directly from its ground-state expectation value,

$$\mathcal{O} = \frac{1}{N-k+1} \sum_{i=1}^{N-k+1} \langle \hat{Q}_i \rangle, \tag{S19}$$

where $\hat{Q}_i$ denotes the corresponding local, string, or cluster Pauli-string operator defined in the main text.

The full-system CE was obtained from the ground-state wavefunction by constructing every reduced density matrix

$$\rho_s = \operatorname{Tr}_{\bar{s}} \left( |\psi_0\rangle\langle\psi_0| \right), \tag{S20}$$

followed by

$$\mathcal{C}(S) = 1 - \frac{1}{2^N} \sum_{s \subseteq S} \operatorname{Tr} \rho_s^2. \tag{S21}$$

No additional approximations were employed.

To facilitate a direct comparison, both $\mathcal{C}(S)$ and $\mathcal{O}$ were linearly normalized to the interval $[0, 1]$,

$$\widetilde{X} = \frac{X - X_{\min}}{X_{\max} - X_{\min}}, \tag{S22}$$

where $X$ denotes either $\mathcal{C}(S)$ or $\mathcal{O}$. The agreement between the two quantities was quantified by the root-mean-square deviation,

$$\mathrm{RMSD} = \sqrt{\frac{1}{M}\sum_{i=1}^{M}\left(\widetilde{\mathcal{C}}_i - \widetilde{\mathcal{O}}_i\right)^2}, \tag{S23}$$

where $M$ is the number of sampled parameter points.

[1] T. Kato, *Perturbation Theory for Linear Operators*, 2nd ed., Classics in Mathematics (Springer, Berlin, Heidelberg, 1995).
[2] M. Reed and B. Simon, *Methods of Modern Mathematical Physics. Vol. IV: Analysis of Operators*, Methods of Modern Mathematical Physics, Vol. 4 (Academic Press, New York, 1978) pp. xv+396.
[3] J. L. Beckey, N. Gigena, P. J. Coles, and M. Cerezo, Computable and operationally meaningful multipartite entanglement measures, Phys. Rev. Lett. **127**, 140501 (2021).
[4] I. Peschel, Calculation of reduced density matrices from correlation functions, Journal of Physics A: Mathematical and General **36**, L205 (2003).
[5] I. Peschel and V. Eisler, Reduced density matrices and entanglement entropy in free lattice models, Journal of Physics A: Mathematical and Theoretical **42**, 504003 (2009).



\end{document}